\documentclass[a4paper,12pt]{article}

\usepackage{amsthm}
\usepackage{amsmath}
\usepackage{amssymb}
\usepackage{latexsym}
\usepackage{float}
\usepackage{diagbox} 
\usepackage{threeparttable}
\usepackage[textwidth=18cm,textheight=20cm]{geometry}
\usepackage[square,sort,comma,numbers]{natbib}

\newtheorem{theorem}{Theorem}[section]
\newtheorem{proposition}[theorem]{Proposition}
\newtheorem{lemma}[theorem]{Lemma}
\newtheorem{corollary}[theorem]{Corollary}
\newtheorem{remark}[theorem]{Remark}
\newtheorem{definition}[theorem]{Definition}
\newtheorem{example}[theorem]{Example}
\newtheorem{assumption}[theorem]{Assumption}
\newcommand{\ud}{\mathrm{d}}
\newcommand{\eqdefr}{=\mathrel{\mathop:}}
\newcommand{\eqdefl}{\mathrel{\mathop:}=}
\newcommand{\Q}{\mathbb{Q}}
\newcommand{\FF}{\mathcal{F}}
\newcommand{\f}{\mathbb{F}}
\newcommand{\A}{\mathcal{A}}
\newcommand{\B}{\mathcal{B}}
\newcommand{\M}{\mathcal{M}}
\newcommand{\K}{\mathcal{K}}
\newcommand{\C}{\mathcal{C}}
\newcommand{\s}{\hat{s}}
\newcommand{\tha}{\hat{\theta}}
\newcommand{\LL}{\mathcal{L}}
\newcommand{\R}{\mathbb{R}}
\newcommand{\1}{\textbf{1}}
\newcommand{\N}{\mathbb{N}}
\newcommand{\PP}{\mathbb{P}}
\newcommand{\Cov}{\mathbf{Cov}}
\newcommand{\Var}{\mathbf{Var}}
\newcommand{\Corr}{\mathbf{Corr}}

\newcommand{\ackname}{Acknowledgements}

\newcommand{\e}{\varepsilon}
\DeclareMathOperator{\inter}{int}
\DeclareMathOperator{\conv}{conv}
\DeclareMathOperator{\supp}{supp}
\DeclareMathOperator{\relint}{ri}
\DeclareMathOperator{\interior}{int}

\usepackage{graphicx}

\newcommand{\F}{\mathscr{F}}

\author{\normalsize{$\text{Hasanjan Sayit}$}  \\
\footnotesize{ $^2\text{Xi'Jiao Liverpool University, Suzhou, China}$ }}

\date{May 7, 2024}

\begin{document}
\title{Weak convergence implies convergence in mean within GGC}
\date{\today}

\maketitle



\abstract{We prove that weak convergence within generalized gamma convolution (GGC) distributions implies convergence in the mean value. We use this fact to show the robustness of the expected utility maximizing optimal portfolio 
under exponential utility function when return vectors are modelled by hyperbolic distributions.
}

\vspace{0.1in}
\vspace{0.1in}

\vspace{0.1in}

\section{Introduction}
The paper \cite{rasonyi-hasan} gives closed form expressions for the expected utility maximizing optimal portfolios when the returns vector of risky assets follow hyperbolic distributions. The market model considered in this paper contains $d+1$ assets with one risk-free asset with interest rate $r_f$ and $d$ risky assets with return vector given by 
\begin{equation}\label{one}
X\overset{d}{=}\mu+\gamma Z+\sqrt{Z}AN,
\end{equation}
where $\mu \in \R^d$ is location parameter, $\gamma \in \R^d$ controls the skewness, $Z\sim G$ is a non-negative mixing random variable,  $A\in \R^{d\times d}$ is a symmetric and positive definite  $d\times d$ matrix of real numbers, and $N\sim N(0, I)$ is a $d-$dimensional  Gaussian random vector with identity co-variance matrix $I$ in $\R^d\times \R^d$ and $N$ is independent from the mixing distribution $Z$. 

The mixing distribution $Z$ in the model (\ref{one}) can be any non-negative random variable. If $Z$ is a non-negative random variable with 
finitely many values then $X$ is called a mixture of Normal random variables (or vectors). If $Z$ follows 
generalized inverse Gaussian (GIG) distribution with the density function
\begin{equation}\label{gign}
 f_{GIG}(x; \lambda, a, b)= (\frac{b}{a})^{\lambda}\frac{1}{K_{\lambda}(ab)}x^{\lambda-1}e^{-\frac{1}{2}(a^2x^{-1}+b^2x)}1_{(0, +\infty)}(x), \end{equation}
where $K_{\lambda}(x)$ denotes the modified Bessel function of third kind with index $\lambda$, then the model (\ref{one}) becomes generalized hyperbolic (GH) distributions. We refer to \cite{Hammerstein_EAv_2010} (Section 1.2, Chapter 1) for further details, especially for the allowed range of parameters $a, b, \lambda$,  of this class of distributions.

The class of GIG distributions belong to the class of generalized gamma convolution (GGC) random variables, see Proposition 1.23 of \cite{Hammerstein_EAv_2010} and Proposition 1.5 of \cite{Mark-Yor-GGC} for this. A positive random variable $Z$ is a GGC, if its Laplace transform takes the following form
\begin{equation}
\mathcal{L}_Z(s)=:Ee^{-sZ}=e^{-\tau s-\int_0^{\infty}log(1+s/x)\nu(dx)},   
\end{equation}
for some  $\tau\geq 0$ called the drift coefficient and some positive measure $U$ called the Thorin measure that satisfy
\begin{equation}\label{vv}
 \int_0^1|log(x)|\nu(dx)<\infty, \; \; \; \int_1^{\infty}\frac{1}{x}\nu(dx)<\infty.   
\end{equation}
The class GGC of distributions are infinitely divisible and self-decomposable, see \cite{Halgreen-1979} and \cite{Barndorff-Halgreen-1977} for these. A random variable $Z$ is a GGC if and only if its Laplace transform  $\mathcal{L}_Z(s)$ is a hyperbolically completely monotone (HCM) function. 

The HCM property for a function $f: (0, +\infty)\rightarrow (0, +\infty)$ means that the function 
\[
f(s_1s_2)f(s_1/s_2)
\]
is completely monotone (CM) as a function of $s=s_1+s_2$ for every $s_1>0$. A positive random variable $Z$ is called HCM if its probability density function $f$ is a HCM.
The paper \cite{Bondesson-1990} proved that the class of HCM distributions belong to GGC. The class HCM is a proper subset of GGC. All the positive $\alpha-$stable random variables, $S_{\alpha}, \alpha\in (0, 1),$ belong to GGC and  a positive $\alpha-$stable random variable is in HCM if and only if $\alpha\le 1/2$, see \cite{Bosch-Simon-2016}. The class HCM is closed under multiplication and division of independent random variables. Also for any $Z\in HCM$ one has $Z^{p}\in HCM$ for every real number $p$ with absolute value $|p|\geq 1$. These facts show that $S_{\alpha}^p$ is in GGC for any $\alpha \in (0, 1/2)$ and $|p|\geq 1$. Other examples of GGC random variables include log-normal models and Pareto distributions.

The motivation of this paper is to show the robustness of the utility maximizing optimal portfolio under exponential utility function when return vectors are modelled by distributions of the form (\ref{one}). For more details of the expected utility maximization problem we refer to \cite{rasonyi-hasan}. To give a short review, in this paper the wealth that corresponds to portfolio weight $x$ on the risky assets is given by
\begin{equation}\label{wealth}
\begin{split}
W(x)=&W_0[1+(1-x^T1)r_f+x^TX] \\
=&W_0(1+r_f)+W_0[x^T(X-\1 r_f)]
\end{split}
\end{equation}
 and the investor's problem is
\begin{equation}\label{L2}
\max_{x\in \R^d}\; EU(W(x)).\\
\end{equation}
According to Proposition 2.17 of the paper \cite{rasonyi-hasan}, a regular solution for the utility maximizing portfolio is give by  
\begin{equation}\label{themainn}
x^{\star}=\frac{1}{aW_0}\Big [\Sigma^{-1}\gamma -q_{min}\Sigma^{-1}(\mu-\1 r_f)\Big ],    
\end{equation}
for some
\begin{equation}\label{32qq}
 q_{min}\in \arg min_{\theta \in (-\tha, \tha)}Q(\theta),
 \end{equation}
where $\tha=:\sqrt{\frac{\mathcal{A}-2\s}{\mathcal{C}}}$ and $\s$ is the IN of the mixing distribution $Z$, see the Definition 2.4 of \cite{rasonyi-hasan} for this. Here $\mathcal{A}, \mathcal{C}, \mathcal{B},$ are given as in equation (32) of \cite{rasonyi-hasan} (here we assume that the model (\ref{one}) is such that $\mu-\1r_f\neq 0$ as in Remark 2.2 of \cite{rasonyi-hasan} and this guarantees $\mathcal{C}\neq 0$). The function $Q(\theta)$ is given as
\begin{equation}\label{H}
Q(\theta)=e^{\mathcal{C}\theta}\mathcal{L}_Z\Big [\frac{1}{2}\mathcal{A}-\frac{\theta^2}{2}\mathcal{C} \Big ].   
\end{equation}
The optimal portfolio (\ref{themainn}) is the utility maximizing optimal portfolio  under the exponential utility function $U(x)=-e^{-ax}, a>0$, and the notation $W_0$ denotes the initial wealth of the investor.

The closed form (\ref{themainn}) of the expected utility maximizing optimal portfolio is obtained without introducing any assumption on the mixing distribution $Z$ in the paper \cite{rasonyi-hasan}. The mixing distribution $Z$ can be a not integrable random variable and we only need to know the IN of $Z$ to be able to write down the utility maximizing optimal portfolio under exponential utility.

While the formula (\ref{themainn}) is convenient in practical applications, one needs to address the issue of robustness of this optimal portfolio with respect to model parameters in the model (\ref{one}). In practice the parameters of the NMVM models are estimated based on Expectation-Maximization algorithm (EM) or other statistical procedures and such estimation procedures  usually give some errors. For this reason it is important to study the robustness of the optimal portfolio 
(\ref{themainn}) with respect to  model parameters $\mu, \gamma, A, Z$ in (\ref{one}). 

To further clarify this point, assume the model (\ref{one}) is the true model for the return vector of risky assets and  assume EM-algorithm  or other statistical estimation procedures lead into a different model
\begin{equation}\label{onee}
X_e\overset{d}{=}\mu_e+\gamma_e Z_e+\sqrt{Z_e}A_eN_n,
\end{equation}
instead of (\ref{one}). Now, assume the differences of $\mu$ with $\mu_e$, $\gamma$ with $\gamma_e$, $A$ with $A_e$, and $Z$ with $Z_e$ are \emph{small} in some norms. Then we would like to show that the optimal portfolio $x^{\star}_e$ based on the model (\ref{onee}) is close to the optimal portfolio $x^{\star}$ obtained based on the model (\ref{one}) in the Euclidean norm. 

To illustrate this robustness issue of the optimal portfolio in (\ref{themainn}) with an example, lets   consider the case of GH models 
$X\sim GH_n(\lambda, \alpha, \beta, \delta, \mu, \Sigma)$. Assume the true parameters are $\lambda, \alpha, \beta, \delta, \mu, \Sigma$ and however EM-algorithym or other statistical procedures lead into a different model $X_e\sim GH_n(\lambda_e, \alpha_e, \beta_e, \delta_e, \mu_e, \Sigma_e)$ with estimated parameters $\lambda_e, \alpha_e, \beta_e, \delta_e, \mu_e, \Sigma_e$. Our main concern is to examine if the optimal portfolio $x^{\star}$ in (\ref{themainn}) for the model $X$ in (\ref{one}) is close to the corresponding optimal portfolio $x^{\star}_e$ for the model $X_e$ in (\ref{onee}) if the parameters $\lambda_e, \alpha_e, \beta_e, \delta_e, \mu_e, \Sigma_e$ are close to the true parameter set $\lambda, \alpha, \beta, \delta, \mu, \Sigma$ in some norms. 

Out of these five parameters $\mu$ and $\beta$ are vectors in $\R^n$ and the Euclidean norm can be used to measure the distance for these parameters. The parameter $\Sigma$ is an $n\times n$ matrix and the Hilbert-Schmit norm for matrices can be used as a measure for distance for this parameter. The other parameters show up in the mixing distribution GIG in the form $\lambda, \delta, \sqrt{\alpha^2-\beta^T\Sigma \beta}$. We don't introduce a measure for closedness for each parameter $\lambda, \delta, \sqrt{\alpha^2-\beta^T\Sigma \beta}$ of the mixing distribution. But we need to use a distance to measure closedness of the laws of two different mixing distributions $Z_1$ and $Z_2$ in the nodel (\ref{one}). We have many options for this and here we review few of them. The Fortet-Mourier distance between laws of random variables is
\begin{equation}
 d_{FM}(Z_1, Z_2)=\sup_{|h|_{\infty}\le 1, |h'|_{\infty}\le 1}|E(h(Z_1))-Eh(Z_2)|, 
 \end{equation}
where $h$ represents continuous functions,  $|\cdot|_{\infty}$ is sup norm within the class of continuous functions, and $h'$ is first order derivative of $h$. It is well known that this distance metrize the convergence in law, i.e., $d_{FM}(Z_n, Z)\rightarrow 0$ if and only if $ d_{FM}(Z_n, Z)\rightarrow 0$ as $n\rightarrow \infty$. Another distance is Kolmogorov's distance
\begin{eqnarray*}
d_{Kol}(Z_1, Z_2)=\sup_{x\in \R}|F_1(X)-F_2(X)|,    
\end{eqnarray*}
where $F_1, F_2$ are commutative distribution functions of $Z_1$ and $Z_2$ respectively. This is a stronger distance than the FM distance in the sense that $d_{Kol}(Z_n, Z)\rightarrow 0$ implies  $ d_{FM}(Z_n, Z)\rightarrow 0$. Another popular distance which is stronger than the Kolmogorov distance is the total variation distance
\begin{equation}
d_{TV}(Z_1, Z_2)=sup_{B\in \mathcal{B}}|P_1(B)-P_2(B)|,    
\end{equation}
where $\mathcal{B}$ is the sigma algebra of Bore sets and $P_1$ and $P_2$ are distribution functions of $Z_1$ and $Z_2$ respectively. By the Scheffe's theorem we have
$d_{TV}(Z_1, Z_1)=\frac{1}{2}\int_0^{+\infty}|f_1(x)-f_2(x)|dx,$ where $f_1$ and $f_2$ are probability density functions of $Z_1$ and $Z_2$. All these distances can be used to measure the closedness of the mixing distributions. However if the mixing distributions are from the class $GGC$ then convergence in law is equivalent to convergence in total variation norm as Lemma \ref{lemWTV} below shows.

 The above discussions motivates us to  measure closedness of the mixing distributions by the weakest  distance $d_{FM}$ that metrize the weak convergence. We will focus our discussions on robustness issue for optimal portfolios that are regular only. The reason is  if an optimal portfolio $x^{\star}$ is irregular, then there exists a sequence of portfolios $x_n^{\star}$ with  $|x_n^{\star}-x^{\star}|\rightarrow 0$ such that $EU(W(x_n^{\star}))=-\infty$ while $ |EU(W(x^{\star}))|<\infty$. This means that a slight mis-specification $\tilde{x}^{\star}$ of the optimal portfolio $x^{\star}$, which can result in from mis-specification of the model parameters in (\ref{one}), can lead to an expected utility that equals to negative infinity while  the \emph{true} optimal portfolio $x^{\star}$ has finite expected utility. This makes the discussion of the robustness issue at irregular solutions  meaningless.
 
 In this paper we use the following notations. For any vectors $x=(x_1, x_2, \cdots, x_d)^T$ and $y=(y_1, y_2, \cdots, y_d)^T$ in $\R^d$, where the superscript $T$ stands for the transpose of a vector, $<x, y>=x^Ty=\sum_{i=1}^dx_iy_i$ denotes the scalar product of the vectors $x$ and $y$, and $|x|=\sqrt{\sum_{i=1}^dx_i^2}$ denotes the Euclidean norm of the vector $x$. For any matrix $A$ we use $|A|_{HS}=\sqrt{\sum_{i=1, j=1}^d|A_{ij}|^2}$ to denote the Hilbert-Schmidt norm of a matrix. We use the notation $\overset{w}{\rightarrow}$ to denote weak convergence of random variables. We sometimes use the short hand notation $X\sim N(\mu+\gamma z, z\Sigma ) \circ G$ for (\ref{one}). $\R$ denotes the set of real numbers and $\R_+=[0, +\infty)$ denotes the set of non-negative real-numbers. Following the same notations of \cite{Mark-Yor-GGC}, $\mathcal{J}$ denotes the family of infinitely divisible random variables on $\R_+$, $\mathcal{S}$ denotes the set of self-decomposable random variables on $\R_+$, and $\mathcal{G}$ denotes the class of generalized gamma convolutions (GGCs) on $\R_+$ that will be introduced later. The Laplace transformation of any distribution $G$ is denoted by $\mathcal{L}_G(s)=\int e^{-sy}G(dy)$.

The paper is organized as follows. In Section 2 below we show that weak convergence within GGC implies convergence of the mean values for the integrable GGC random variables. In Section 3 we use this fact to show the robustness of the optimal portfolios for the problem (\ref{L2}).

\section{Weak convergence  within GGC}

Our discussions about robustness of the optimal portfolio will be focused on models of the form (\ref{one}) with mixing distribution $Z$ that belong to the GGC class of random variables. The class of GGC models include the class of GIG models as discussed in Proposition 1.23 of \cite{Hammerstein_EAv_2010}. The class GGC of distributions appeared in Thorin's work and it includes popular models like log-normal distributions, Pareto distributions, and all positive  stable random variables. They are closed in weak limits, addition of independent random variables,   and multiplication of independent random variables, see \cite{Bondesson-ProbTheory} for further details.

Before we start our discussions, we first write down the definition of GGC random variables, see \cite{Bondesson-ProbTheory, Bondesson-TheProb} for more details.

\begin{definition}\label{defGGC} A GGC is a probability distribution on $[0, \infty)$ with  Laplace transformation of the form $\phi(s)=exp\{-\tau s-\int_0^{\infty}log(1+\frac{s}{t})\nu(dt)\}$, where $\nu(dt)$ is a nonnegative measure on $(0, \infty)$ and it satisfies (\ref{vv}),  and $\tau$ is a non-negative number which is called left extremity of the GGC random variable. The pair $(\tau, \nu)$ is called generator of a GGC random variable. When $\tau=0$, we call the associated random variable a GGC without a drift.
\end{definition}

\begin{remark}\label{rem3.6} Let $Z$ be a GGC with generating pair $(\tau, \nu)$. Define a GGC random variable $\bar{Z}$ with LT given by $\bar{\phi}(s)=e^{-\int_0^{\infty}log(1+\frac{s}{t})\nu(dt)}$. Then $Z\overset{d}{=}\bar{Z}+\tau$ or equivalently $Z-\tau\overset{d}{=}\bar{Z}$.
\end{remark}

 Before we discuss the robustness problem that is stated above, we first collect few properties of the GGC distributions in the following Lemma.   

\begin{lemma}\label{2.9} Let $Z$ be a nondegenerate GGC random variable with generating pair $(\tau, \nu)$. Let $\s$ be the IN of $Z$.
Then $\s$ is a finite number and the measure $\nu$  satisfies $\nu([0, -\s])=0$. We have $\mathcal{L}_{Z}(s)=Ee^{-sZ}=e^{-\tau s-\int_{-\s}^{\infty}log(1+\frac{s}{z})\nu(dz)}$ and
there exists a deterministic strictly positive and  decreasing (almost surely) function $h(s)$ on $[-\s, +\infty)$ such that $Z-\tau \overset{d}{=}\int_{-\s}^{+\infty}h(s)d\gamma_s$, where $\gamma_s$ is a standard gamma
subordinator with L\'evy measure $\frac{e^{-x}}{x}dx, x>0$.
\end{lemma}
\begin{proof} Due to Remark \ref{rem3.6}, it is sufficient to consider GGC random variables with zero drift. Therefore below we assume $Z$ is a GGC with generating pair $(0, \nu)$. Recall the Laplace transformation 
$\mathcal{L}_{Z}(s)=Ee^{-sZ}=e^{-\int_0^{\infty}log(1+\frac{s}{z})\nu(dz)}
$
of $Z$ from the 
definition 1.0 of \cite{Mark-Yor-GGC} at page 3.51.  By the definition of $s_0$, the integral  $\int_0^{\infty}log(1+\frac{s}{z})\nu(dz)$ is finite
for all $s> \s$. This means that one should have 
$1+\frac{s}{z}>0$ for all $s>\s$. But this is true only if $z>-\s$. This implies that the Thorin measure $\nu$ needs to assign zero measure to $[0, -\s]$. We can't have $\s=-\infty$ as this would imply $\nu([0, +\infty))=0$,  a contradiction for the non-degenerancy of  $Z$. From part 2. of Proposition 1.1 of \cite{Mark-Yor-GGC}, we have $Z=\int_{-\s}^{+\infty}h(s)d\gamma_s$ with $h(s)=1/F^{-1}_{\nu}(x)$, where $F^{-1}_{\nu}$ 
is the right continuous inverse of $F_{\nu}(x):=\int_{(-\s, x]}\nu(dx)$ on $[-\s, +\infty)$. The function $F_{\nu}(x)$ is a finite valued increasing function and therefore its right continuous inverse $F_{\nu}^{-1}$ is also finite valued and increasing function. Therefore $h(s)$ is a decreasing function. Now if $h(s)=0$ on $[-\s, s')$ for some $s'>-\s$, then $\nu([0, s'])=0$ which contradicts with the definition of $\s$. Therefore $h(s)>0$ almost surely on $[-\s, +\infty)$.
\end{proof}

Next we state the following continuity theorem which is useful for our discussions.

\begin{theorem}\label{conti} Let $\{Z_n\}$ be a family from GGC random variables with generating pairs $\{(\tau_n, \nu_n)\}$. Assume $Z_n$ converges weakly to a distribution $Z$. Then $Z$ is also a GGC with a generating  pair $(\tau, \nu)$. We have $\nu_n$
converges weakly to $\nu$,
$\tau=\lim_{M\rightarrow \infty}\lim_{n\rightarrow \infty}[\tau_n+\int_M^{\infty}\frac{1}{x}\nu_n(dx)]$, and $\lim_{\delta\rightarrow 0}\lim_{n\rightarrow \infty}\int_0^{\delta}(\ln t)\nu_n(dt)=0$.
\end{theorem}
\begin{proof} See page 35 of \cite{Bondesson-lecturenotes}, also see \cite{Thorin1977OnTI} and Theorem 1.22 of \cite{Hammerstein_EAv_2010} for this.
\end{proof}

\begin{remark}\label{38}
Before we state our next result we make few observations on the LT of a GGC random variable first. According to Proposition 1 of \cite{Bondesson-TheProb},
a function $\phi(s)$ is a LT of a GGC random variable iff $\phi(s)$ is analytic in $C/(-\infty, 0]$ and without zeros and $\phi(0)=1$ and $Im[\phi'(s)/\phi(s)]\geq 0$ for $Im(s)> 0$. This means that $\phi(s)$ is differentiable for any $s>0$ and 
\begin{equation}\label{LTGGC}
\phi'(s)/\phi(s)=-\tau-\int_0^{\infty}\frac{1}{t+s}\nu(dt),    
\end{equation}
where $(\tau, \nu)$ is the generator of a GGC random variable $Z$ with LT equals to $\phi(s)$. Now assume $EZ<\infty$, then 
from $\phi(s)=Ee^{-sZ}$ we have $\phi'(s)=-E[Ze^{-sZ}]$ for all $s>0$. By dominated convergence theorem we have 
$\lim_{s\rightarrow 0+}\phi'(s)=-EZ$. Now from (\ref{LTGGC}) we obtain 
\begin{equation}\label{LTGGC1}
EZ=\tau+\int_0^{+\infty}\frac{1}{t}\nu(dt).  \end{equation}
Note that here we used monotone convergence theorem for the integral in the right-hand-side of (\ref{LTGGC}) to obtain (\ref{LTGGC1}). Now recall that a Thorin measure satisfies $\int_1^{+\infty}\frac{1}{t}\nu(dt)<\infty$ and $\int_0^{1}|log t|\nu(dt)<\infty$. The relation (\ref{LTGGC1}) shows that when $EZ<\infty$ it also satisfies $\int_0^1\frac{1}{t}\nu(dt)<\infty$.
\end{remark}

From this we immediately obtain the following result.
\begin{lemma}\label{contilem} Let $\{Z_n\}$ be a family from GGC and assume $Z_n$ converges weakly to $Z$ and assume $EZ_n<\infty, EZ<\infty$. Let $(\tau_n, \nu_n)$ be the generator of $Z_n$ for each $n$ and let $(\tau, \nu)$ be the generator of $Z$. Define $g^{(n)}(\delta)=:\int_0^{\delta}\frac{1}{t}\nu_n(dt)$ for each $n\geq 1$ and assume
\begin{equation}\label{condition}
 g(\delta)=:\lim_{n\rightarrow \infty}g^{(n)}(\delta)   
\end{equation}
exists and finite for each $\delta \in [0, c]$ for some $c>0$. Then  $g(\delta)=q(\delta)=:\int_0^{\delta}\frac{1}{t}\nu(dt)$ on $[0, c]$ and at the same time we have $EZ_n\rightarrow EZ$.
\end{lemma}
\begin{proof} By the continuity theorem \ref{conti}, $\nu_n$ weakly converges to $\nu$. Therefore we have 
$g^{(n)}(b)-g^{(n)}(a)=\int_a^b\frac{1}{t}\nu_n(dt)\rightarrow \int_a^b\frac{1}{t}\nu(dt)=q(b)-q(a)$ for any $c\geq b\geq a>0$. Since $g(\delta)$ is point-wise limit of $g^{(n)}(\delta)$ we have $g(b)-g(a)=q(b)-q(a)$ for any $c\geq b\geq a>0$. Being a point-wise limit of  monotone decreasing functions, $g(\delta)$ is a decreasing function on $[0, c]$. The function  
$q(\delta)$ is also a decreasing function. Therefore the limits of $g(a)$ and $q(a)$ when $a\rightarrow 0$ exists.
By taking the limit when $a\rightarrow 0$ to the equation $g(b)-g(a)=q(b)-q(a)$ we obtain $g(b)-g(0)=q(b)$.
It remains now to show that $g(0)=0$. Assume $g(0)>0$ and below we show that this leads into a contradiction. Since $g^{(n)}(\delta)$ is a decreasing function with $g_n(0)=0$ for each fixed $n$, we can find a convergent sequence $\delta_n\rightarrow 0$ such that $g^{(n)}(\delta_n)\le \frac{g(0)}{2}$. But $g_n(\delta_n)\rightarrow g(0)$ and this is a contradiction. Therefore we have $g(\delta)=\int_0^{\delta}\frac{1}{t}\nu(dt)$.
Next we show that $EZ_n\rightarrow EZ$. To see this observe that 
\begin{equation}\label{ezen}
EZ_n=\tau_n+\int_M^{\infty}\frac{1}{t}\nu_n(dt)+\int_{\delta}^M\frac{1}{t}\nu_n(dt)+\int_0^{\delta}\frac{1}{t}\nu_n(dt)
\end{equation}
for any $0<\delta \le c<M<\infty$. By the continuity theorem \ref{conti} we have $\tau=\lim_{M\rightarrow \infty}\lim_{n\rightarrow \infty}[\tau_n+\int_M^{\infty}\frac{1}{x}\nu_n(dx)]$. Also we have $\lim_{M\rightarrow \infty}\lim_{n\rightarrow \infty}\int_{\delta}^M\frac{1}{t}\nu_n(dt)=\lim_{M\rightarrow \infty}[\lim_{n\rightarrow \infty}\int_{\delta}^M\frac{1}{t}\nu_n(dt)]=\lim_{M\rightarrow \infty}\int_{\delta}^M\frac{1}{t}\nu(dt)=\int_{\delta}^{\infty}\frac{1}{t}\nu(dt)$. Now taking the limit $\lim_{M\rightarrow \infty}\lim_{n\rightarrow \infty}$ to the both sides of (\ref{ezen})  we obtain $EZ_n \rightarrow \tau+\int_{\delta}^{\infty}\frac{1}{t}\nu(dt)+\int_0^{\delta}\frac{1}{t}\nu(dt)$=EZ. This ends the proof.
\end{proof}

\begin{lemma}\label{Corr} Let $\{Z_n\}$ be a family from GGC and assume $Z_n$ converges weakly to $Z$ and assume $EZ_n<\infty, EZ<\infty$. Let $\s$ be the IN  of $Z$. If $\s \neq 0$, then we have $EZ_n\rightarrow EZ$. Equivalently, if the Thorin measure $U$ of $Z$ satisfies $U([0, d])=0$ for some $d>0$, then we have $EZ_n\rightarrow EZ$. 
\end{lemma}
\begin{proof} Let $\s_n$ and  $\nu_n$ denote the IN  and the Thorin measure of $Z_n$ for each $n\geq 1$ respectively. Let $\nu$ denote the Thorin measure of $Z$. From Lemma \ref{3.13} below, we have $\s_n\rightarrow \s$. Therefore there exists positive integer $n_0$ such that $|\s_n|\geq \delta=:|\s|/2$ for all $n\geq n_0$. Define 
$g^{(n)}(\theta)=:\int_0^{\theta}\frac{1}{t}\nu_n(dt)$ and $g(\theta)=:\int_0^{\theta}\frac{1}{t}\nu(dt)$ for any $\theta\geq 0$. By Lemma \ref{2.9}, when $\theta < \delta$ we have $\nu_n([0, \theta])=0$ for all $n\geq n_0$ and  $\nu([0, \theta])=0$. This shows that $g^{(n)}(\theta)=0, \; g(\theta)=0$ for all $\theta<\delta$. Therefore the condition (\ref{condition}) is trivially satisfied with $c=\frac{\delta}{2}$. From these analysis observe that  $U([0, d])=0$ for some $d>0$ if and only if the IN $\s$ of $Z$ satisfies $\s \neq 0$. This completes the proof.
\end{proof}

\begin{lemma}\label{contiboun} Let $\{Z_n\}$ be a family from GGC and assume $Z_n$ converges weakly to $Z$ with  $EZ<\infty$. Let $(\tau_n, \nu_n)$ be the generator of $Z_n$ for each $n$ and let $(\tau, \nu)$ be the generator of $Z$. Assume $\{Z_n\}$
has a sub-sequence $\{Z_{n_k}\}$ such that $EZ_{n_k}<\infty$ for all $k\geq 1$ and the sequence $\{EZ_{n_k}\}_{k\geq 1}$ converges to a finite number, then $EZ_{n_k}\rightarrow EZ$. Especially, if, in addition,  $EZ_n$ is a bounded sequence  then we have $EZ_{n}\rightarrow EZ$.
\end{lemma}
\begin{proof} Since any sub-sequence of the weakly convergent sequence converges weakly, it is sufficient to prove if  $EZ_n<\infty$ for all $n$ and if $EZ_n$ converges to a finite number then $EZ_n\rightarrow EZ$. For this it is sufficient to show $g^{(n)}(\delta)=:\int_0^{\delta}\frac{1}{t}\nu_n(dt)$ converges point-wise to a finite valued function $g(\delta)$ on $[0, c]$ for some $c>0$ by the above Lemma \ref{contilem}. Fix any number $c>0$ and observe that 
\begin{equation}\label{ezen}
EZ_n=\tau_n+\int_M^{\infty}\frac{1}{t}\nu_n(dt)+\int_{\delta}^M\frac{1}{t}\nu_n(dt)+\int_0^{\delta}\frac{1}{t}\nu_n(dt)
\end{equation}
for any $0<\delta \le c<M<\infty$. By the continuity theorem the limit 
$\lim_{M\rightarrow \infty}\lim_{n\rightarrow}[\tau_n+\int_M^{\infty}\frac{1}{t}\nu_n(dt)]$ exists. Also the limit $\lim_{M\rightarrow \infty}\lim_{n\rightarrow \infty}\int_{\delta}^M\frac{1}{t}\nu_n(dt)=\lim_{M\rightarrow \infty}[\lim_{n\rightarrow \infty}\int_{\delta}^M\frac{1}{t}\nu_n(dt)]=\lim_{M\rightarrow \infty}\int_{\delta}^M\frac{1}{t}\nu(dt)=\int_{\delta}^{\infty}\frac{1}{t}\nu(dt)$ exists for each fixed $0<\delta<c$. By the assumption the sequence $EZ_n$  converges to a  finite limit. Therefore from (\ref{ezen}) the limit of $\lim_{n\rightarrow \infty}\int_0^{\delta}\frac{1}{t}\nu_n(dt)$ exists for each $0<\delta\le c$ and it is finite valued. This completes the proof. To see the second claim of the Lemma 
it is sufficient to show any convergent sub-sequence of $\{EZ_n\}$ converges to the same number $EZ$. Let $EZ_{n_k}$
be a convergent sub-sequence of $\{EZ_n\}$. Since by assumption $\{EZ_n\}$ is a bounded sequence, $EZ_{n_k}$ converges to a finite number. Therefore from the claim in the first half of the Lemma we have $EZ_{n_k}\rightarrow EZ$. This completes the proof for the Lemma.
\end{proof}

The above Lemmas \ref{contilem} and \ref{contiboun} give some sufficient conditions for the convergence of the mean when a sequence of GGC converges weakly to a GGC. Below we use these Lemmas to show that actually weak convergence in the family of GGC random variables imply convergence of the corresponding mean values. Namely we will show that  
the condition (\ref{condition}) in Lemma \ref{contilem} and the boundedness assumption 
of the expected values in Lemma \ref{contiboun} can be dropped. Before we prove this result we need some preparations.

A Gamma distribution $\xi\sim G(\alpha, \beta)$ has density function $g(x)=x^{\alpha-1}e^{-x/\beta}/(\beta^{\alpha}\Gamma(\alpha))$. A right-shifted Gamma distribution is given by $\eta=\xi+\tau$, where $\xi \sim G(\alpha, \beta)$ and $\tau \geq 0$ (see section 1.5 at page 28 of \cite{Hammerstein_EAv_2010} for more details). We use the notation $G(\alpha, \beta, \tau):=Law(Y)$ to denote the right-shifted gamma distributions. In our discussions we need to consider independent sequences $G(\alpha_i, \beta_i, \tau_i)$ of right-shifted Gamma distributions. We write $G(\alpha_i, \beta_i, \tau_i)=\xi_i+\tau_i$, where $\{\xi_i\}_{1\le i\le n }$ is then a sequence of independent gamma random variables with probability density functions
\begin{equation}
f_i(x_i)=f_i(x_i; \alpha_i, \beta_i)=x^{\alpha_i-1}e^{-x_i/\beta_i}/[\beta_i^{\alpha_i}\Gamma(\alpha_i)],\; \alpha_i>0, \; \beta_i>0,\;  x_i>0.    
\end{equation}
We have $\xi_i\sim \beta_i\bar{\xi}_i$ where $\bar{\xi}_i\sim f_i(x_i; \alpha_i, 1)$ and $\{\bar{\xi}_i\}_{1\le i\le n}$ are mutually independent. With these we have  $\bar{Z}=:\bar{\xi}_1+\bar{\xi}_2+\cdots+\bar{\xi}_n\sim f(x; \alpha, 1)$, where $\alpha=\sum_{i=1}^n\alpha_i$. The distribution of $Z=:\xi_1+\xi_2+\cdots+\xi_n$ is not known in closed form. Let $g(x)$ denote the probability density function of $Z$. Denote $\beta_m=:\min_{1\le i\le n}\beta_i$ and for notational simplicity, without loosing any generality,  we can  assume that $\beta_m=\beta_1$. Then, when not all of $\{\beta_i\}$ are equal to each other,  the paper \cite{MOSCHOPOULOS} in its equation (2.12) gives the following bound for $g$
\begin{equation}\label{62}
g(x)\le [C/(\beta_{m}^{\rho}\Gamma(\rho))]x^{\rho-1}e^{-x(1-v)/\beta_m},    
\end{equation}
where
\begin{equation}
C=\prod_{i=1}^n(\beta_m/\beta_i)^{\alpha_i},\; \rho=\sum_{i=1}^n\alpha_i, \; v=\max_{2\le i\le n}(1-\beta_m/\beta_i).    
\end{equation}
\begin{remark} Since for two independent gamma random variables $\xi_1\sim f(x; \alpha_1, \beta_1)$ and $\xi_2\sim f(x; \alpha_2, \beta_2)$ with $\beta=:\beta_1=\beta_2$, we have $\xi_1+\xi_2\sim f(x; \alpha_1+\alpha_2, \beta)$, in (\ref{62}) we can assumed that $\beta_1=\beta_m<\min_{2\le i\le n}\beta_i$.
Also we observe that the numbers $C$ and $\nu$ in (\ref{62}) are bounded numbers.
\end{remark}

In our next result, we will consider sequences 
\begin{equation}\label{Gcon}
\bar{Z}_n=\prod_{i=1}^{k_n} \ast G(\alpha_i^{(n)}, \beta_i^{(n)}, \tau_i^{(n)}), \; \; n\geq 1,
\end{equation}
of finite convolutions of right-shifted gamma distributions, where $k_n, n\geq 1,$ are positive integers. If we denote $\tau^{(n)}=\sum_i^{k_n}\tau_i^{(n)}$, then we have 
\begin{equation}\label{64}
\bar{Z}_n\sim \tau^{(n)}+\xi^{(n)}_1+\xi_2^{(n)}+\cdots+\xi_{k_n}^{(n)}, \; \;   \xi_i^{(n)}\sim f(x; \alpha^{(n)}_i, \beta^{(n)}_i),\;  1\le i\le k_n.
\end{equation}
We denote 
\begin{equation}\label{Zgc}
Z_n=:\xi^{(n)}_1+\xi_2^{(n)}+\cdots+\xi_{k_n}^{(n)}, n\geq 1   
\end{equation}
and 
\begin{equation}\label{65}
C^{(n)}=:\prod_{i=1}^{k_n}(\beta_m^{(n)}/\beta_i^{(n)})^{\alpha_i^{(n)}},\; \rho^{(n)}=:\sum_{i=1}^{k_n}\alpha_i^{(n)}, \; v^{(n)}=\max_{2\le i\le k_n}(1-\beta_m^{(n)}/\beta_i^{(n)}),    
\end{equation}
where $\beta_m^{(n)}=\min_{1\le i\le k_n}\beta^{(n)}_i=\beta_1^{(n)}$ (again $\beta_1^{(n)}$ is assumed to be minimum of $\{\beta_i^{(n)}\}_{1\le i\le k_n}$). Let $g_n(x)$ denote the probability density function of $Z_n$ for each $n\geq 1$. Then assuming that not all of $\beta_i^{(n)}, 1\le i\le k_n$ are equal to each other, we have 
\begin{equation}\label{gn}
 g_n(x)\le [C^{(n)}(\beta_{m}^{(n)})^{-\rho^{(n)}}/\Gamma(\rho^{(n)})]x^{\rho^{(n)}-1}e^{-x(1-v^{(n)})/\beta_m^{(n)}}.   
\end{equation}

Next we prove the following Lemma.
\begin{lemma}\label{3.13} Let $\bar{Z}_n$ be a sequence of finite convolutions of right-shifted gamma distributions. If $\bar{Z}_n$ converges weakly to a non-degenerate random variable $Z$ with $EZ<\infty$,  we have $E\bar{Z}_n \rightarrow EZ$. 
\end{lemma}
\begin{proof} Let $(\tau, U)$ denote the generating pair for $Z$. We assume that the 
sequence $\bar{Z}_n$ is given by (\ref{Gcon}). We denote by $U^{(n)}$ the Thorin measure of $\bar{Z}_n$
for each $n\geq 1$. To prove the claim in the Lemma it is sufficient, by Lemma \ref{contiboun},  to prove that $\{E\bar{Z}_n\}$ is a bounded sequence. Below we divide the proof into two parts a) and b).

(a) For arbitrarily fixed positive numbers $0<\delta<K<+\infty$, we define
\begin{equation}
\begin{split}
&I^{(n)}_{0,\delta}=\{i: \frac{1}{\beta_i^{(n)}}\le \delta,\; 1\le i\le k_n\}, \\
&I^{(n)}_{\delta, K}=\{i: K\geq \frac{1}{\beta_i^{(n)}}> \delta,\;  1\le i\le k_n\}, \\
&I^{(n)}_{K, +\infty}=\{i: \frac{1}{\beta_i^{(n)}}> K,\; 1\le i\le k_n\},
\end{split}
\end{equation}
and we write $\bar{Z}_n= \tau^{(n)}+\tilde{Z}_1^{(n)}+\tilde{Z}_2^{(n)}+\tilde{Z}_3^{(n)}$, where $\tilde{Z}_1^{(n)}=\sum_{i\in I_{0, \delta}^{(n)}}\xi_i^{(n)}, \tilde{Z}_2^{(n)}=\sum_{i\in I_{\delta, K}^{(n)}}\xi_i^{(n)},$ and $\tilde{Z}_3^{(n)}=\sum_{i\in I_{K, +\infty}^{(n)}}\xi_i^{(n)}$. We write the Laplace transformation of $Z$ as $ \mathcal{L}_Z(s)=\mathcal{L}_1\mathcal{L}_2\mathcal{L}_3$ with 
\begin{equation}\label{69}
 \mathcal{L}_1=e^{\int_0^{\delta}log(\frac{1}{1+s/t})U(dt)} , \mathcal{L}_2=e^{\int_{\delta}^{K}log(\frac{1}{1+s/t})U(dt)}, \mathcal{L}_3=e^{\tau+\int_{K}^{+\infty}log(\frac{1}{1+s/t})U(dt)},   
\end{equation}
and denote by $\tilde{Z}_1$ the GGC random variable that corresponds to $\mathcal{L}_1$, and by $\tilde{Z}_2$ the GGC random variable that corresponds to $\mathcal{L}_2$, and similarly by $\tilde{Z}_3$ the GGC that corresponds to $\mathcal{L}_3$ (see for a similar idea in the proof of Theorem 3.1 at page 130 of \cite{Bondesson-scandinal-actuarial-j}). The Laplace distribution of $Z$ is the multiplication of three Laplace transformations and therefore $Z$ is equal in distribution to 
the independent sum of  $\tilde{Z}_1, \tilde{Z}_2$, and $\tilde{Z}_3$, i.e., $Z\sim \tilde{Z}_1+\tilde{Z}_2+\tilde{Z}_3$. We also have that $\tilde{Z}_1^{(n)}$ converges weakly to $\tilde{Z}_1$, $\tilde{Z}_2^{(n)}$ converges weakly to $\tilde{Z}_2$, and $\tau_n+\tilde{Z}_3^{(n)}$ converges weakly to $\tilde{Z}_3$. To see this, observe that the Laplace distribution of $\tilde{Z}_1^{(n)}$ does not contribute to $\mathcal{L}_2$ and $\mathcal{L}_3$ in the limit due to the restriction in the set $I_{0, \delta}^{(n)}$ above. Therefore its Laplace distribution need to converge $\mathcal{L}_1$, which in turn implies  $\tilde{Z}_1^{(n)}$ converges weakly to $\tilde{Z}_1$. Similar arguments hold for the other two.

To show $\{E\bar{Z}_n\}$ is a bounded sequence, it is sufficient to show all of  $E\tilde{Z}_1^{(n)}, E\tilde{Z}_2^{(n)},$ and $\tau_n+E\tilde{Z}_3^{(n)}$ are bounded sequences. The rest of the proof is devoted to show this claim. First observe that
\begin{equation}
E(\tau_n+\tilde{Z}_3^{(n)})=\tau_n+\int_K^{+\infty}\frac{1}{t}U^{(n)}(dt).    
\end{equation}
By the continuity Theorem \ref{conti} we have 
\[
\lim_{K\rightarrow +\infty}\lim_{n+\infty}[\tau_n+\int_K^{+\infty}\frac{1}{t}U^{(n)}(dt)]=\tau<\infty.
\]
Therefore for a fixed number $\epsilon>0$ there exists positive integer $n_0$ and positive number $K_0$ such that $\tau_n+\int_{K_0}^{+\infty}\frac{1}{t}U^{(n)}(dt)$ lies in the interval $[\tau-\epsilon, \tau+\epsilon]$ when $n\geq n_0$. For the simplicity of notations we denote $K_0$  by $K$ and the $\{\tilde{Z}_3^{(n)}\}$ (these are defined to be $\tilde{Z}_3^{(n)}=\sum_{i\in I_{K_0, +\infty}^{(n)}}\xi_i^{(n)}$ in fact and we use $K$ instead of $K_0$ for notational simlicity) are defined as above. We conclude that  $E(\tau_n+Z_3^{(n)})$ is a bounded sequence of $n$. Also we have $EZ_2^{(n)}=\int_{\delta}^K\frac{1}{t}U^{(n)}(dt)$ and as $Z_n$ converges weakly to $Z$, the measure $U^{(n)}$ converges weakly to $U$ by the continuity Theorem \ref{conti}. Therefore 
$EZ_2^{(n)}=\int_{\delta}^{K}\frac{1}{t}U^{(n)}(dt)\rightarrow \int_{\delta}^{K}\frac{1}{t}U(dt)<\infty$. From this we conclude that $EZ_2^{(n)}$ is also a bounded sequence of $n$. Next we show that $EZ_1^{(n)}$ is a bounded sequence of $n$ also and we do this in part b) below.

b) Recall that $\tilde{Z}_1^{(n)}=\sum_{i\in I_{0, \delta}^{(n)}}\xi_i^{(n)}$. We have $E\tilde{Z}^{n}_1=\int_0^{\delta}\frac{1}{t}U^{(n)}(dt)$. Here we can't apply the same idea that we have used for $\{Z_2^{(n)}\}$ and $\{Z_3^{(n)}\}$ in part a) above as we don't know if $\int_0^{\delta}\frac{1}{t}U^{(n)}(dt)\rightarrow \int_0^{\delta}\frac{1}{t}U(dt)$ holds true. However we have an 
upper bound as in (\ref{62}) for the density functions of finite gamma convolutions and we use this fact to show that $EZ_1^{(n)}$ is a bounded sequence. First observe that
$\int_0^{\delta}U^{(n)}(dt)=\sum_{i\in I_{0, \delta}^{(n)}}\alpha_i^{(n)}\rightarrow \int_0^\delta U(dt)<\infty$ as $U^{(n)}$ weakly converges to $U$. Therefore $\rho_1^{(n)}=:\sum_{i\in I_{0, \delta}^{(n)}}\alpha_i^{(n)}$ is a bounded sequence of $n$. We have 
$EZ_1^{(n)}=\sum_{i\in I_{0, \delta}^{(n)}}\alpha_i^{(n)}\beta_i^{(n)}.$ Now if $\{\beta_i^{(n)}, i\in I^{(n)}_{0, \delta}\}_{n\geq 1}$ is a uniformly  bounded family of $n$, then clearly we have $EZ_1^{(n)}$ is a bounded sequence of $n$. Therefore we exclude this case from our discussions below and we assume that $\{\beta_i^{(n)}, n\geq 1, i\in I^{(n)}_{0, \delta}\}$ is an unbounded family of $n$.

Below for the sake of notations we assume that the family $\{Z_1^{(n)}\}$ is given by (\ref{Zgc}) and we denote this family, with abuse of notations,  by $\{Z_n\}$. So we have $Z_n$ converges weakly to $\tilde{Z}_1$ in part a) above. The corresponding probability density functions have the upper bounds as in (\ref{gn}). We use the same notations $\rho^{(n)}, \beta_m^{(n)},$
as in (\ref{65}). We also define $\beta_M^{(n)}=\max_{1\le i\le k_n}\beta_i^{(n)}$. We divide the family $\{Z_n\}_{n\geq 1}$ (namely the family $\{Z_1^{(n)}\}$) into two disjoint sets $\{Z_n\}_{n\geq 1}=\{Z_n'\}_{n\geq 1}\cup \{Z_n^{''}\}_{n\geq 1}$, where each member of 
$\{Z_n'\}$ is finite gamma convolutions with not all $\{\beta'_i\}$ are equal to each other and the family $\{Z_n^{''}\}$ is such that each member of it is finite gamma convolutions with equal $\beta^{''}_i$. Both of $Z_n^{'}$ and $Z_n^{''}$ (being sub-sequences of $Z_n$, more precisely of $Z_1^{(n)}$) converge weakly to $\tilde{Z}_1$. 

First consider the sequence $Z_n^{''}$. It is clear that, in fact, each $Z_n^{''}$ are single gamma random variables with $Z_n^{''}\sim G(\rho{''}_n, \beta^{''}_n)$ for some $\rho_n^{''}$ and $\beta_n^{''}$. Clearly $\rho_n^{''}$ is a sub-sequence of $\rho_1^{(n)}$ defined in the above paragraph. As such $\rho_n^{''}$ is a bounded sequence. Now assume $\rho_n^{''}\rightarrow 0$. Let $U_n^{''}$ denote the Thorin measure of $Z_n^{''}$ for each $n\geq 1$. Then since $Z_n^{''}$ weakly converges to $\tilde{Z}_1$, we have $\int_0^{\delta}U(dt)=\lim_{n\rightarrow +\infty}\int_0^{\delta}U_n^{''}=\lim_{n\rightarrow +\infty}\rho_n^{''}=0$, which shows $U([0, \delta])=0$. But if this is true then we have $E\bar{Z}_n\rightarrow EZ$ already by Lemma \ref{Corr}. So we assume $\inf_n\rho_n^{''}>0$ below. 
Now if $EZ_n^{''}$ is an unbounded sequence then we should have  $\beta_n^{''}\rightarrow +\infty$ as $EZ_n^{''}=\rho_n^{''}\beta_n^{''}$ and $\rho_n^{''}$ is a bounded sequence from below and above as explained above. Then the sequence of probability density functions $f_n^{''}(x)=x^{\rho_n^{''}-1}e^{-x/\beta_n^{''}}/[(\beta_n^{''})^{\rho_n^{''}}\Gamma(\rho_n^{''})]$ of $Z_n^{''}$ converges to zero almost surely. But at the same time  $f_n^{''}(x)$ should converge to the density function of $\tilde{Z}_1$ by Theorem 4.1.5 of \cite{Bondesson-lecturenotes} implying that the probability density function of $\tilde{Z}_1$ is zero, a contradiction. Therefore $EZ_n{''}$ is a bounded sequence.

Now it remains to show that $EZ_n^{'}$ is a bounded sequence. Below, with another abuse of notation, we denote the family $\{Z_n'\}$ by $\{Z_n\}$ and show that $EZ_n$ is a bounded sequence. The probability density function of $Z_n$ is denoted by $g_n(x)$ for each $n$. We have 
\begin{equation}\label{important}
 g_n(x)\le [C^{(n)}(\beta_{m}^{(n)})^{-\rho^{(n)}}/\Gamma(\rho^{(n)})]x^{\rho^{(n)}-1}e^{-x(1-v^{(n)})/\beta_m^{(n)}}.   
\end{equation}
In (\ref{important}), the family $\{\rho^{(n)}\}_{n\geq 1}$ is uniformly bounded as explained above. Since the Thorin measure of $\tilde{Z}_1^{(n)}$ has support in the interval $[0, \delta]$, we have 
$\frac{1}{\beta_m^{(n)}}\le \delta$ for all $n$ and this gives a lower bound $\beta_m^{(n)}\geq \frac{1}{\delta}$ for the family $\{\beta_m^{(n)}\}_{n\geq 1}$. Also we can assume that the family $\{\beta_m^{(n)}\}_{n\geq 1}$ is uniformly bounded. If not, since the Thorin measure $U^{(n)}|_{[0, \delta]}$ (the restriction of the Thorin measure of $\bar{Z}_n$ to $[0, \delta]$)  has support in $[\frac{1}{\beta_M^{(n)}}, \frac{1}{\beta_m^{(n)}}]$, we have $U^{(n)}([0, \delta])\rightarrow 0$ as $n\rightarrow 0$. This implies that $U([0, \delta])=0$ and we are reduced to the case of Lemma \ref{Corr} and so we have  $E\bar{Z}_n\rightarrow EZ$ trivially. Therefore we assume $\{\beta_m^{(n)}\}_{n\geq 1}$ is uniformly bounded. Also, as explained above,  we can assume that $\beta_M^{(n)}$ is an unbounded sequence (this corresponds to the Thorin measure $U$ has support in any close neighborhoods of $0$) as if it is bounded sequence then the Thorin measure $U$ satisfies $U([0, \inf_n{(1/\beta_M^{(n)})}])=0$ with $\inf_n{(1/\beta_M^{(n)})}>0$ and again we are reduced to the case of Lemma \ref{Corr}.

Now we look at $v^{(n)}=\max_{2\le i\le k_n}(1-\beta_m^{(n)}/\beta_i^{(n)})$ in (\ref{important}). As explained above we can assume that $\{\beta_m^{(n)}\}$ is a bounded family and $\{\beta_M^{(n)}\}$ is an unbounded family. Therefore there exists $n_0$ such that for all $n\geq n_0$ we have $v^{(n)}\le \frac{1}{2}$. Therefore when $n\geq n_0$ we have $(1-v^{(n)})/\beta_m^{(n)}\geq \beta_m^{(n)}\geq \frac{1}{2B}$, where $B$ is any fixed  upper bound for $\{\beta_m^{(n)}\}$. Also observe that 
\[
\ln C^{(n)}=\beta_m^{(n)}\{\sum_{i\in I_{0, \delta}^{(n)}}\alpha_i^{(n)}\}-\sum_{i\in I_{0, \delta}^{(n)}}\alpha_i^{(n)}\ln \beta_i^{(n)}.
\]
Since $\sum_{i\in I_{0, \delta}^{(n)}}\alpha_i^{(n)}\ln \beta_i^{(n)}=\int_0^{\delta}(\ln t)U_n(dt)\rightarrow 0$ by the continuity Theorem \ref{conti}, we can conclude that $\{C^{(n)}\}_{n\geq 1}$ is a bounded sequence. Therefore for all $n\geq n_0$  we have
\begin{equation}\label{733}
g_n(x)\le \bar{C}x^{\bar{\rho}-1}e^{-\bar{B}x}, \end{equation}
where $\bar{B}=\frac{1}{2B}$, $\bar{\rho}$ is any fixed upper bound for $\rho^{(n)}$, and $\bar{C}$ is any fixed  upper bound for the family  $\{C^{(n)}(\beta_{m}^{(n)})^{-\rho^{(n)}}/\Gamma(\rho^{(n)})\}_{n\geq 1}$. Next, we show that (\ref{733}) implies that $EZ_n$ is a bounded sequence. To see this, recall that $Z_n$ converges weakly to $\tilde{Z}_1$ and $E\tilde{Z}_1<\infty$. Let $g(x)$ denote the probability density function of $\tilde{Z}_1$. By Theorem 4.1.5 of \cite{Bondesson-lecturenotes}, the sequence $g_n(x)$ converges point-wise to $g(x)$.
We have the following
\begin{equation}\label{2L}
|EZ_n-E\tilde{Z}_1|\le \int_0^{+\infty}x|g_n(x)-g(x)|dx=J_1^{(n)}(L)+J_2^{(n)}(L),
\end{equation}
where $J_1^{(n)}(L)=:\int_0^Lx|g_n(x)-g(x)|dx$ and $J_2^{(n)}(x)=:\int_L^{+\infty}x|g_n(x)-g(x)|dx$ and $L$ can be  any positive number but we require $L>1$ . If we can show that both $\{J_1^{(n)}(L)\}_{n\geq 1}$ and $\{J_2^{(n)}(L)\}_{n\geq 1}$ are bounded family then from (\ref{2L}) we can conclude that $\{EZ_n\}$ is a bounded family. For each fixed $L>1$, we have $J_1^{(n)}(L)$ converges to zero by the dominated convergence Theorem. For $J_2(L)$ we have
\[
J_2(L)\le \int_L^{+\infty}xg_n(x)dx+\int_L^{+\infty}xg(x)dx.
\]
We have $\int_L^{+\infty}xg(x)dx\le E\tilde{Z}_1\le EZ<\infty$ for all $L> 1$. For $\int_L^{+\infty}xg_n(x)dx$ we have 
\[
\int_L^{+\infty}xg_n(x)dx\le \bar{C}\int_L^{+\infty}x^{\bar{\rho}}e^{-\bar{B}x}dx,
\]
due to (\ref{733}) whenever $n\geq n_0$. By applying integration by parts multiple times  we can obtain $\int_L^{+\infty}x^{\bar{\rho}}e^{-\bar{B}x}dx\le Q(L)e^{-\tilde{B}L}$, where $Q(L)$ is a polynomial of $L$ (here we need to use the requirement $L>1$). But $Q(L)e^{-\tilde{B}L}\rightarrow 0$ as $L\rightarrow +\infty$ for any polynomial $Q(L)$ of $L$. Therefore for a given finite number $M>0$, we have a positive number $L_0>1$ such that $Q(L)e^{-\bar{B}L}\le M$ for all $L\geq L_0$.
Then $J_2^{(n)}(L_0)\le EZ+M$ for all $n\geq n_0$. We also have $J_1^{(n)}(L_0)\rightarrow 0$ as explained above. Therefore there exists 
a positive integer $n_0'$ such that $J_1^{(n)}(L_0)\le M$ for all $n\geq n_0'$. Then for all $n\geq \max\{n_0, n_0'\}$ we have 
\[
|EZ_n-E\tilde{Z}_1|\le 2M+EZ.
\]
This shows that $\{EZ_n\}$ is a bounded sequence as $E\tilde{Z}_1<\infty$. This completes the proof.
\end{proof}

\begin{remark}\label{3.14} The above Lemma \ref{3.13} shows that if the members of the sequence $\{Z_n\}$ are finite convolutions of right shifted gamma random variables, then the weak convergence of $Z_n$ to a non-degenerate random variable $Z$ with $EZ<\infty$  implies the convergence of the mean value, i.e, $EZ_n \rightarrow EZ$. Here, unlike Lemma  \ref{contiboun}, we don't have to require the boundedness of the sequence $EZ_n$. As the proof of this Lemma shows the properties of the finite gamma convolutions play important role for the proof of this Lemma. We wish to show a similar result for the general class of GGC random variables. For this, we need to use a result from the paper \cite{rao}: let $Z_n$ be a sequence of random variables with the Laplace transformations $\mathcal{L}_{Z_n}(s)$ and $Z$ be a random variable with Laplace transformation $\mathcal{L}_Z(s)$. Then, according to \cite{rao},  the sequence  $\mathcal{L}_{Z_n}(s)$ converges to $\mathcal{L}_Z(s)$ point-wise on some interval $(c, d)$ if and only if $Z_n$ converges weakly to $Z$ and $\sup_n\mathcal{L}_{Z_n}(s)<\infty$ for all $s\in (c, d)$. Here $(c, d)$ can be any open interval. Also we need to use the fact that if a sequence of  monotone continuous functions $g_n(x)$ converge point-wise to a continuous function $g(x)$ on a compact interval $[c, d]$, then the convergence is uniform on $[c, d]$. 
\end{remark}

\begin{theorem}\label{keytheorem} Let $Z_n$ be a sequence of non-degenerate random variables from GGC with $EZ_n<\infty$ for each $n\geq 1$. Assume $Z_n$
weakly converges to a non-degenerate random variable $Z$ with $EZ<\infty$. Then we have $EZ_n\rightarrow EZ$. 
\end{theorem}
\begin{proof} Let $\{Z^{(n)}_k\}_{k\geq 1}$ denote a sequence of finite  convolutions of right-shifted  gamma distributions that converges weakly to $Z_n$ for each $n$. Let $g_k^{(n)}(s)=:\mathcal{L}_{Z^{(n)}_k}(s)$ denote the Laplace transformations of $Z_k^{(n)}$ respectively for all $k\geq 1, n\geq 1$.
Let $g^{(n)}(s)=:\mathcal{L}_{Z_n}(s)$ denote the Laplace transformation of $Z_n$ for each $n\geq 1$. Let $g(s)=:\mathcal{L}_{Z}(s)$ denote the Laplace transormation of the limit random variable $Z$. For each fixed $n$, we have $g_k^{(n)}(s)$ converges point-wise to $g^{(n)}(s)$ on the compact interval $[0, 1]$ due to weak convergence of $Z_k^{(n)}$ to $Z^{(n)}$. Since all of  $g_k^{(n)}(s), g^{(n)}(s)$ are decreasing functions on $[0, 1]$, the convergence is uniform. Also, by following a similar argument, we have $g^{(n)}(s)$ converges uniformly to $g(s)$ on $[0, 1]$. Now, from above Lemma \ref{3.13}, we have $EZ_k^{(n)}\rightarrow EZ_n$ for each fixed $n\geq 1$. Therefore  for each fixed $n\geq 1$, we can pick a finite gamma convolution $Z_{k_n}^{(n)}$ with the property
\begin{equation}\label{68}
 |EZ_n-EZ_{k_n}^{n}|\le \frac{1}{n}, \; \; \sup_{s\in [0, 1]}|\mathcal{L}_{Z_n}(s)-\mathcal{L}_{Z_{k_n}^{(n)}}(s)|\le \frac{1}{n}, 
\end{equation}
 at the same time. We now show that the sequence $Z_{k_n}^{(n)}$ converges weakly to $Z$. Let $\epsilon>0$ be an arbitrary small number.
 Since $Z_n$ weakly converges to $Z$, there exists a positive integer $n_0$ such that $\sup_{s\in [0, 1]}|\mathcal{L}_{Z}(s)-\mathcal{L}_{Z_{n}}(s)|\le \epsilon$ and $\frac{1}{n}<\epsilon$ at the same time for all $n\geq n_0$. Then for all $n\geq n_0$ we have 
 \begin{equation}
 \begin{split}
  \sup_{s\in [0, 1]}|\mathcal{L}_{Z}(s)-\mathcal{L}_{Z_{k_n}^{(n)}}(s)|\le &  \sup_{s\in [0, 1]}|\mathcal{L}_{Z}(s)-\mathcal{L}_{Z_n}(s)|+\sup_{s\in [0, 1]}|\mathcal{L}_{Z_n}(s)-\mathcal{L}_{Z_{k_n}^{(n)}}(s)|,\\
  \le& 2\epsilon.
  \end{split}
 \end{equation}
 Since $\epsilon>0$ is an arbitrary small number, we conclude that $\mathcal{L}_{Z_{k_n}^{(n)}}(s)$ converges to $\mathcal{L}_{Z}(s)$ point-wise on $(0, 1)$. Then by Remark \ref{3.14}, the sequence of finite gamma convolutions $Z_{k_n}^{(n)}$ converges weakly to $Z$. Now if we have $EZ_n\rightarrow +\infty$, then due to (\ref{68}) we have $EZ_{k_n}^{(n)}\rightarrow +\infty$ also. But this contradicts with Lemma \ref{3.13} above.  Therefore $\{EZ_n\}$ is a bounded sequence. Then $EZ_n\rightarrow EZ$ follows from  Lemma \ref{contiboun}.
\end{proof}

\begin{lemma}\label{lem3.10} Let $\{Z_n\}$ be a family from GGC with corresponding generating pairs $\{(\tau_n, \nu_n)\}$. Assume $Z_n$ converges weakly to $Z$ with generating pair $(\tau, \nu)$. Let $Z_n-\tau_n\overset{d}{=}\int_0^{\infty}h_n(s)d\gamma_s$, $Z-\tau\overset{d}{=}\int_0^{\infty}h(s)d\gamma_s$ denote  Wiener-Gamma representations with unique increasing functions $h_n(s), h(s)$. Then $h_n(s)\rightarrow h(s)$ pointwise.
\end{lemma}\label{hhn}
\begin{proof} By the continuity Theorem \ref{conti}, $\nu_n$ weakly converges to $\nu$. Denote 
$F_{\nu_n}(x)=\int_{(0, x]}\nu_n(dy)$ and $F_{\nu}(x)=\int_{(0, x]}\nu(dy)$ and denote by 
$F_{\nu_n}^{-1}, F_{\nu}^{-1}$ their respective right-continuous inverses. From part 2 of Proposition 1.1 of \cite{Mark-Yor-GGC} we have 
$h_n(s)=1/F_{\nu_n}^{-1}$ and $h(s)=1/F_{\nu}^{-1}$. Now from Theorem \ref{conti}, weak convergence implies $F_{\nu_n}\rightarrow F_{\nu}$ pointwise and this in turn implies $F_{\nu_n}^{-1}\rightarrow F_{\nu}^{-1}$. Therefore we have $h_n(s)\rightarrow h(s)$ pointwise.
\end{proof}

Recall that our goal is to discuss the robustness problem of the optimal portfolio in (\ref{themainn}). The optimal portfolio in this theorem involves the Laplace transformation of the mixing distribution. Therefore, we first need to study the properties of the Laplace transformation of the GGC random variables. Especially, we would like to study the relation of weak convergence with the convergence of the corresponding Laplace transformations within the class of GGC distributions. First recall the classical result that a sequence $Z_n$ converges to $Z$ in distribution, i.e., $Z_n\overset{w}{\rightarrow} Z$, if and only if $Ef(Z_n)\rightarrow Ef(Z)$ for any bounded and continuous function $f$. This result clearly does not imply that the Laplace transformations $\mathcal{L}_{Z_n}(s)$ of the random variables $Z_n$ converge to the Laplace transformation $\mathcal{L}_{Z}(s)$ of the random variable $Z$ under the condition that $Z_n\overset{w}{\rightarrow} Z$ as the functions $f(x)=e^{-sx}$ are not bounded functions when $s<0$. In our setting all the random variables in GGC are non-negative
and therefore $\mathcal{L}_{Z_n}(s)\rightarrow \mathcal{L}_{Z}(s)$ for all $s\geq 0$ as long as 
$Z_n\overset{w}{\rightarrow} Z$. But when $s<0$ such result does not immediately follow as then the function $e^{-sx}$ is no longer bounded on $(0, +\infty)$ when $s<0$ as mentioned above. But interestingly we will show that $Z_n\overset{w}{\rightarrow} Z$ implies  $\mathcal{L}_{Z_n}(s)\rightarrow \mathcal{L}_Z(s)$ whenever $\mathcal{L}_Z(s)<\infty$ as long as $Z_n, Z$ are within the class GGC of distributions. First, in the next simple Lemma we state some useful facts on the GGC random variables.

\begin{lemma}\label{KeyLemma} Let $Z_n, Z$ be a family from GGC with respective generators $(\tau_n, \nu_n)$ and $(\tau, \nu)$. Assume $EZ_n<\infty, EZ<\infty$. Let $\s_n, \s$ denote their corresponding IN. Let $b_n>0$ be any sequence of real numbers with $b_n\rightarrow b>0$, $b_n<|\s_n|$ for all $n\geq 1$, and $b<|\s|$. Then if $Z_n\overset{w}{\rightarrow}Z$ we have $Ee^{b_nZ_n}\rightarrow Ee^{bZ}$. 
\end{lemma}
\begin{proof} Since $b_{n}\rightarrow b$ we have $b_{n}Z_{n}\overset{w}{\rightarrow}bZ$. Since $e^{x}$ is a continuous function we have $\eta_n=:e^{b_{n}Z_{n}}\overset{w}{\rightarrow }\eta=:e^{bZ}$. By the definitions of $b, \; b_n$ we have $Ee^{b_{n}Z_{n}}<\infty, Ee^{bZ}<\infty$ for all $n$. Exponentials of GGC are again GGC (see \cite{Bondesson-TheProb} for this). So $\eta_n, \eta$ are GGC. The claim now follows from Theorem \ref{keytheorem}.
\end{proof}

\begin{remark}\label{keyrem} From the above Lemma \ref{KeyLemma}, it can be  seen  that if a real number $s>0$ satisfies $s<|\s_n|$ for all $n\geq 1$ and $s<|\s|$ then $Ee^{sZ_n}\rightarrow Ee^{sZ}$ whenever $Z_n, Z$ satisfy the hypothesis of the Lemma \ref{KeyLemma}.
\end{remark}

\begin{lemma}\label{3.13} Assume all of $Z_n, Z$ are nondegenerate GGC random variables with corresponding generating pairs $(\tau_n, \nu_n), (\tau, \nu)$. Let $\s_n$ and $\s$ denote their corresponding  IN respectively. If $Z_n\overset{w}{\rightarrow} Z$, then we have $\s_n\rightarrow \s$.
\end{lemma}
\begin{proof} Due to Remark \ref{rem3.6}, we can assume $\tau_n=0, \tau=0$. First we  show $\{|\s_n|\}$ is a bounded sequence. Assume $\{|\s_n|\}$ has an unbounded sub-sequence. We show that this leads into a contradiction. Without loss of any generality we assume $|\s_n|\rightarrow \infty$. From Lemma \ref{2.9} we have $Z_n=\int_{|\s_n|}^{+\infty}h_{n}(s)d\gamma_s=\int_{0}^{+\infty}h_n (s)1_{[|\s_n|\;\; \infty)}(s)d\gamma_s$ and $Z=\int_{|\s|}^{+\infty}h(s)d\gamma_s=\int_0^{+\infty}h(s)1_{[|\s|, \infty)]}(s)d\gamma_s$ 
for some decreasing deterministic functions $h^{(n)}(s)$ and $h(s)$. Since $Z_n$ converges weakly to $Z$, by lemma \ref{lem3.10} we should have $h_n (s)1_{[|\s_n|\;\; \infty)}(s)\rightarrow h(s)1_{[|\s|, \infty)]}(s)$ point-wise almost surely. By Lemma \ref{2.9} we have $h(s)>0$ on $[|\s|, \infty)$. But this is not possible if $|\s_n|\rightarrow \infty$ while $|\s|<\infty$. Therefore $\{|\s_n|\}$ is a bounded sequence. Next we show $\s_n\rightarrow \s$. For this it is sufficient to show that any convergent sub-sequence of $\{\s_n\}$ converges to $\s$. To show this, 
without loss of any generality,  we assume that $\s_n\rightarrow s'$ and show that $\s=s'$. We first assume $|s'|<|\s|$ and find a contradiction. From the definitions of the numbers $|\s_n|$ (recall $Ee^{|\s_n|Z_n}=+\infty$) we can claim the existence of real numbers $\delta_n>0$ with $\lim_{n\rightarrow \infty}\delta_n=0$ such that $Ee^{(|\s_n|-\delta_n)Z_n}\rightarrow +\infty$. We let $b_n=|\s_n|-\delta_n, n\geq 1,$ and $b=|s'|$. Then these numbers $\{b_n, b\}$ satisfy the conditions of Lemma \ref{KeyLemma}. Therefore we should have $sup_nEe^{b_nZ_n}<\infty$ and this contradicts with $Ee^{(|\s_n|-\delta_n)Z_n}\rightarrow +\infty$.  Now we assume $|s'|>|\s|$ and find a contradiction. By Lemma \ref{2.9} we have $Z_n=\int_{|\s_n|}^{\infty}h_n(s)d\gamma_s=\int_0^{\infty}h_n(s)1_{[|\s_n|, \infty)]}(s)d\gamma_s$,
$Z=\int_{|\s|}^{\infty}h(s)d\gamma_s=\int_0^{\infty}h(s)1_{[|\s|, \infty)]}(s)d\gamma_s$ and $h(s)>0$ on $[|\s|, \infty)$. Since $Z_n$ converges weakly to $Z$, by Lemma \ref{lem3.10}
we have $h_n(s)1_{[|\s_n|, \infty)]}(s)$ converges pointwise to $h(s)1_{[|\s|, \infty)]}(s)$. But this is not possible if $|\s_n|\rightarrow |s'|>|\s|$ while $h(s)>0$ on 
$[|\s|, \infty)$.

\end{proof}

\begin{lemma}\label{lem215} Assume all of $Z_n, Z$ are nondegenerate GGC random variables. Let $(\tau_n, \nu_n)$ and $(\tau, \nu)$ be their respective generators. Assume $EZ_n<\infty,\;  EZ<\infty$. Then if $Z_n\overset{w}{\rightarrow} Z$, we have
\begin{equation}\label{lapp}
\mathcal{L}_{Z_n}(s)\rightarrow  \mathcal{L}_{Z}(s), \end{equation}
whenever $\mathcal{L}_{Z}(s)<\infty$. 
\end{lemma}
\begin{proof} The claim (\ref{lapp}) is true for any $s\geq 0$ as in this case $e^{-sx}$ are bounded continuous functions on $[0, +\infty)$. Therefore we need to show it for negative $s$. Let $\s_n, \s$ denote the IN  of $Z_n, Z$ respectively. We exclude the case $\s=0$ as in this case $\mathcal{L}_Z(s)= +\infty$ for any $s<0$. Then it is sufficient to assume $0<s<|\s|$ and show that $Ee^{sZ_n}\rightarrow Ee^{sZ}$. Fix such $s$ and denote $\delta=(|\s|-s)/2$. In Lemma \ref{2.9} we showed that $\s_n\rightarrow \s$. 
 Therefore there exists a positive integer $n_0$ such that $|\s_n|>s+\delta$ for all $n\geq n_0$. Note that we have $|\s|> s+\delta$ also. Then from Remark \ref{keyrem} we  have $Ee^{sZ_n}\rightarrow Ee^{sZ}$. This completes the proof.  
\end{proof}

In the next Lemma we show that for non-degenerate GGC random variables with IN number $\s$, if $\s\neq 0$ then necessarily
$\LL_Z(\s)=+\infty$. This fact will be used in the proof of Proposition \ref{prop3.3} below.

\begin{lemma}\label{5.5} Let $\s$ be the IN of a non-degenerate GGC random variable $Z$. If $\s\neq 0$, then we have $\LL_Z(\s)=+\infty$.
\end{lemma}
\begin{proof} First look at the case of a gamma random variable $Z\sim G(\alpha, \beta)$ with shape parameter $\alpha$ and scale parameter $\frac{1}{\beta}$. We have $\mathcal{L}_Z(s)=\frac{1}{(1+\beta s)^{\alpha}}$. In this case the IN is $\s=-\frac{1}{\beta}$ and clearly we have $\mathcal{L}_Z(-\frac{1}{\beta})=+\infty$.
If $Z$ is finite gamma convolution  $Z\sim \sum_{i=1}^nG(\alpha_i, \beta_i)$, we have $\mathcal{L}_{Z}(s)=
\prod \frac{1}{(1+\beta_i s)^{\alpha_i}}$ and in this case the IN  of $Z$ is $\s=\max_{1\le i \le n}\{-\frac{1}{\beta_i}\}$ and one can easily check that $\mathcal{L}_{G}(\s)=+\infty$ in this case also. Now any GGC random variable $Z$ with zero drift is a weak limit of a sequence $Z_n$ of finite gamma convolutions. Denote the IN of $Z$ by $\s$ and the IN number of $Z_n$ by $\s_n$ for each $n\geq 1$. We have 
$\mathcal{L}_{Z_n}(\s_n)=+\infty$ for each $n\geq 1$. Therefore there exists a non-negative sequence of deterministic numbers $\epsilon_n$ with $\epsilon_n\rightarrow 0$ such that $\mathcal{L}_{Z_n}(\s_n+\epsilon_n)=
Ee^{(|\s_n|-\epsilon_n)Z_n} 
 <+\infty$ and 
$\lim_{n\rightarrow \infty}\mathcal{L}_{Z_n}(\s_n+\epsilon_n)=\lim_{n\rightarrow \infty}Ee^{(|\s_n|-\epsilon_n)Z_n}\rightarrow +\infty$. By Lemma \ref{3.13}, we have $|\s_n|-\epsilon_n\rightarrow |\s|$. Now if we assume $\mathcal{L}_Z(\s)=Ee^{|\s|Z}<+\infty$, then by Lemma \ref{KeyLemma} we should have
$Ee^{(|\s_n|-\epsilon_n)Z_n}\rightarrow Ee^{|\s|Z}<\infty$, a contradiction. For a GGC random variable $G$ with generating pair $(\tau, \nu)$, the random variable $G-\tau$ is a GGC with generating pair $(0, \nu)$ and the IN for both $Z$ and $Z-\tau$ are equal to each other. For any $s$ we have $e^{\tau s}\mathcal{L}_Z(s)=\mathcal{L}_{Z-\tau}(s)$. From this we conclude that for any non-degenerate GGC random variable $Z$ we have $\mathcal{L}_Z(\s)=+\infty$ when the IN  of $Z$ is not zero.
\end{proof}

\section{Robustness of the exponential utility maximizing portfolio}

In this subsection,  we address the robustness issue of the optimal portfolio in the paper \cite{rasonyi-hasan} as an application of our results in Section 2 above. First we prove the following Lemma.

\begin{lemma}\label{lemWTV} Assume $Z_n, Z$ are nondegenerate random variables from the class GGC. Then $Z_n\overset{w}{\rightarrow}Z$ if and only if $d_{TV}(Z_n, Z)\rightarrow 0$. Also $Z_n\overset{w}{\rightarrow}Z$ implies $d_{Kol}(Z_n, Z)\rightarrow 0$.
\end{lemma}
\begin{proof} All nondegenerate random variables in GGC have probability density functions and they are unimodel, see part vi) of \cite{Bondesson-ProbTheory} for this (also
see 
the introduction of \cite{Wissem-Thomas}). Then the claim follows from \cite{Reis} (also see page 383 of \cite{Nourdin}). Since the limit distribution $Z$ has probability density function  convergence in law implies convergece in the Kolmogorov distance, a fact that can be derived by using Dini's second theorem (see the introduction of \cite{Nourdin} for this).
\end{proof}

Next we define the models that are necessary for the discussion of robustness. Let $(\mu_n)$ and $(\gamma_n)$ be any family of vectors in $\R^d$. Let $(A_n)$ be a family of symmetric and positive definite  matrices in $\R^d\times \R^d$. Let $(Z_n)$ be a family of non-negative random variables that are independent from $N_n$ (the $d-$dimensional standard normal random variables). We define  the following models
\begin{equation}\label{one-n}
X_n=\mu_n+\gamma_nZ_n+\sqrt{Z_n}A_nN_n.    
\end{equation}
Also we define $\Sigma_n=A_nA_n^T$ and 
\begin{equation}\label{ABCn}
\A_n=\gamma^T_n\Sigma^{-1}_n\gamma_n,\; \C_n=(\mu_n-\1 r_f)^T\Sigma^{-1}_n(\mu_n-\1 r_f),\; \B_n= \gamma^T_n \Sigma^{-1}_n(\mu_n-\1 r_f),  
\end{equation}
for each positive integer $n$. For each $n\geq 1$, the corresponding utility maximization problem for the model (\ref{one-n}) is
\begin{equation}\label{maxn}
\max_{x\in \R^d}EU(W_n(x)),     
\end{equation}
where 
\begin{equation*}
\begin{split}
W_n(x)=&W_0[1+(1-x^T1)r_f+x^TX_n] \\
=&W_0(1+r_f)+W_0[x^T(X_n-\1 r_f)].
\end{split}
\end{equation*}
As discussed at the beginning of Section 3, we assume the model (\ref{one}) is the \emph{true} model and the parameters of the models (\ref{one-n}) converge to the corresponding parameters of the \emph{true} model. Namely we assume the following holds
\begin{equation}\label{conver-ro}
 \mu_n\rightarrow \mu, \; \gamma_n\rightarrow \gamma,\;  A_n\rightarrow A.  
\end{equation}
Denote the solution of (\ref{maxn}) by $x_n^{\star}$ for each $n\geq 1$ and the solution of (\ref{L2}) by $x^{\star}$. In this section, we would like to show that $x_n^{\star}\rightarrow x^{\star}$ under some conditions.

We first prove the following Lemma. Since all the matrices here are symmetric matrix we drop the transpose operator ''T''
in our calculations below. We also drop the symbol ``HS'' in the norm $|\cdot|_{HS}$ whenever there is no confusion arises.
\begin{lemma}\label{2.8} Assume the  $(\mu_n)$ and $(\gamma_n)$ in the models (\ref{one-n}) are convergent sequences of real vectors with limits $\mu$ and $\gamma$ in the model (\ref{one}). Let the $(A_n)$ in (\ref{one-n}) be a sequence of $d\times d$ symmetric positive definite matrices that satisfy $|A_n-A|_{HS}\rightarrow 0$, where $A$ is the matrix in the model (\ref{one}). Then we have 
\begin{equation}
 \A_n\rightarrow \A,\; \C_n\rightarrow \C,\; \B_n\rightarrow \B,   
\end{equation}
where $\A,\; \B,\; \C,$ are given as in (32) of \cite{rasonyi-hasan}. We also have 
\begin{equation}\label{siggam}
\Sigma_n^{-1}\gamma_n\rightarrow \Sigma^{-1}\gamma, \; \; \Sigma_n^{-1}(\mu_n-\1r_f)\rightarrow \Sigma^{-1}(\mu-\1 r_f).
\end{equation}

\end{lemma}
\begin{proof} Note that $\A_n=\gamma^T_n(A_n^T)^{-1}A_n^{-1}\gamma_n=(A_n^{-1}\gamma_n)^TA_n^{-1}\gamma_n$ and  also $\A=\gamma^TA^{-1}(A^{-1})^T\gamma=(A^{-1}\gamma)^T(A^{-1}\gamma)$. Therefore it is sufficient to show $\gamma_n^TA_n^{-1}\rightarrow \gamma^TA^{-1}$ in Euclidean norm. Since $A_n\rightarrow A$, from \cite{Stewart1969OnTC} we have $A_n^{-1}\rightarrow A_n^{-1}$ (inverse of non-singular matrix is a continuous function of the elements of the matrix). We have
\begin{equation*}
\begin{split}
|\gamma^T_nA_n^{-1}-\gamma^TA^{-1}|\le & |\gamma^{T}_n(A_n^{-1}-A^{-1})|+|(\gamma_n^T-\gamma^T)A^{-1}|\\
\le&|\gamma_n| |A_n^{-1}-A^{-1}|_{HS}+|\gamma_n-\gamma| |A^{-1}|_{HS}\\
\le& |\gamma_n-\gamma| |A_n^{-1}-A^{-1}|_{HS}+||\gamma||A_n^{-1}-A^{-1}||_{HS}+|\gamma_n-\gamma||A^{-1}|_{HS}.
\end{split}
\end{equation*}
From this the claim follows. The other cases $\C_n\rightarrow \C, \B_n\rightarrow \B$ can be proved similarly. To show (\ref{siggam}) it is sufficient to show $\Sigma_n^{-1}\rightarrow \Sigma^{-1}$ as the remaining parts of the proof follows from similar arguments as  above. The relation $\Sigma_n^{-1}\rightarrow \Sigma^{-1}$ follows from 
\[
\begin{split}
|\Sigma_n^{-1}-\Sigma^{-1}|=&|A_n^{-1}A_n^{-1}-A^{-1}A^{-1}|=|(A_n^{-1}-A^{-1})A_n^{-1}+A^{-1}(A_n^{-1}-A^{-1})|\\
\le& |(A_n^{-1}-A^{-1})A_n^{-1}|+|A^{-1}(A_n^{-1}-A^{-1})|\\
\le & |(A_n^{-1}-A^{-1})(A_n^{-1}-A^{-1})|+|(A_n^{-1}-A^{-1})A^{-1}|+|A^{-1}(A_n^{-1}-A^{-1})|\\
\le &|(A_n^{-1}-A^{-1})||(A_n^{-1}-A^{-1})|+|(A_n^{-1}-A^{-1})||A^{-1}|
+|A^{-1}||(A_n^{-1}-A^{-1})|.
\end{split}
\]
\end{proof}

Now, for each random variable $Z_n$ in (\ref{one-n}), let
$\mathcal{L}_{Z_n}(s)$ denote its Laplace transformation and let $\s_n$ denote its IN. We define 
\begin{equation}\label{Hn}
Q_n(\theta)=e^{\C_n\theta}\mathcal{L}_{Z_n}\Big [\frac{1}{2}\A_n-\frac{\theta^2}{2}\C_n \Big ], \;\;  \tha_n=:\sqrt{\frac{\A_n-2\s_n}{\C_n}} 
\end{equation}
for each $n\geq 1$. The functions $Q_n(\theta)$ are strictly convex on $(-\tha_n, \tha_n)$ for each $n$ when $\s_n$ is finite with $\LL_Z(\s_n)=+\infty$ or when $\s_n=-\infty$ and it is strictly convex in $[-\tha_n, \tha_n]$ when $\s_n$ is finite and $\LL_Z(\s_n)<\infty$ as will be proved in Lemma 4 of \cite{rasonyi-hasan}. These properties are important for the proof of our result as we shall see.

Below are the optimal portfolios for the problems (\ref{maxn}) and (\ref{L2}) and we write them down here for convenience.
\begin{equation}\label{32qn}
\begin{split}
 q_{min}^{(n)}&=:\arg min_{\theta \in (-\tha_n, \tha_n)}Q_n(\theta),\;\;\;  x^{\star}_n=\frac{1}{aW_0}\Big [\Sigma^{-1}_n\gamma_n -q_{min}^{(n)}\Sigma^{-1}_n(\mu_n-\1 r_f)\Big ],    
\\
 q_{min}&=:\arg min_{\theta \in (-\tha, \tha)} Q(\theta), \;\;\;  x^{\star}=\frac{1}{aW_0}\Big [\Sigma^{-1}\gamma -q_{min}\Sigma^{-1}(\mu-\1 r_f)\Big ].
 \end{split}
 \end{equation}

\begin{proposition}\label{prop3.3}  Consider the model (\ref{one}) and assume $Z$ is a non-degenerate GGC random variable with $EZ<\infty$. Assume the associated problem (\ref{L2}) with the model (\ref{one}) has a regular solution $x^{\star}$. Let (\ref{one-n}) be a sequence of models that satisfy (\ref{conver-ro}). Assume 
$Z_n$ in (\ref{one-n}) are non-degenerate mixing distributions from the class GGC also and $EZ_n<\infty$. Assume for each $n\geq 1$, the problem (\ref{maxn}) has regular solution  $x_n^{\star}$.  Then if $Z_n\overset{w}{\rightarrow} Z$, we have  $|x_n^{\star}-x^{\star}|\rightarrow 0$.   
\end{proposition}
\begin{proof} Let $\s_n$ and $\s$ denote the IN of $Z_n$ and $Z$ respectively. Let $\tha_n=\sqrt{(\A_n-2\s_n)/\C_n}$ and $\tha=\sqrt{(\A-2\s)/\C}$. The assumption on the regularity of $x_n^{\star}$ and $x^{\star}$ imply all of $(-\tha_n, \tha_n)$ and $(\tha, \tha)$ are non-empty open intervals. By Proposition 2.17 of \cite{rasonyi-hasan}, the solutions $x_n^{\star}$ and $x^{\star}$ to the problems (\ref{maxn}) and (\ref{L2})  are given by (\ref{32qn}) respectively. We have $q^{(n)}_{min}\in (-\tha_n, \tha_n)$ and  $q_{min}\in (-\tha, \tha)$ as the solutions $x_n^{\star}$ and $x^{\star}$ are regular by the assumption. Due to lemma \ref{2.8}, we only need to prove $q_{min}^{(n)}\rightarrow q_{min}$. 

First observe that $\{q_{min}^{(n)}\}$ is a bounded sequence as $\{\s_n\}$ is a bounded sequence due to Lemma \ref{3.13} and hence $\tha_n$ is a bounded sequence by Lemma \ref{2.8}. So to prove $q_{min}^{(n)}\rightarrow q_{min}$ it is sufficient to prove that any convergent sub-sequence $q_{min}^{(n_k)}$ converges to the same number $q_{min}$. Without loss of any generality, below we show that if  $q_{min}^{(n)}\rightarrow q$ then $q=q_{min}$. For the simplicity of notations below we write $q_n=:q_{min}^{(n)}$  for all $n\geq 1$.  By Lemma 4.1 of \cite{rasonyi-hasan} we have $q_n\in (-\tha_n, 0)$ and 
$q_{min}\in (-\tha, 0)$. Being a limit of $q_n$ and also since $\tha_n\rightarrow \tha$, we can assume $q\in [-\tha, 0]$ below. We divide the proof into two cases.

Case 1: Assume $\s\neq 0$. Then by Lemma \ref{5.5} we have $\LL_Z(\s)=+\infty$ and therefore $Q(-\tha)=+\infty$. We first show that $Q_n(\theta)\rightarrow Q(\theta)$ for all $\theta\in \bar{\Theta} =:(-\tha, -\sqrt{\frac{\A}{\C}})\cup (-\sqrt{\frac{\A}{\C}}, 0)$.  Denote $\eta_n(\theta)=:\frac{1}{2}\A_n-\frac{\theta^2}{2}\C_n$ and $\eta(\theta)=:\frac{1}{2}\A-\frac{\theta^2}{2}\C$. Here  we singled out the point $\tilde{\theta}=:-\sqrt{\A/\C}$  as $\eta(\tilde{\theta})=0$. It is easy to see that if 
$\theta\in (-\tha_n, 0)$ we have $\s_n<\eta_n(\theta)< \frac{1}{2}\A_n$ and if $\theta\in (-\tha, 0)$ we have $\s<\eta(\theta)< \frac{1}{2}\A$. Observe that $\eta_n(\theta)\rightarrow \eta(\theta)$ for all $\theta \in \R$ due to Lemma \ref{2.8}. Take any $\theta_0\in \bar{\Theta}$. Then either $\eta(\theta_0)>0$ or $\eta(\theta_0)<0$. If $\eta(\theta_0)>0$ then we have $\eta_n(\theta_0)>0$ for all $n$ that are sufficiently large. Since $Z_n\overset{w}{\rightarrow}Z$ and the function $e^{-sx}$ is bounded continuous function of $x\geq 0$ for each $s>0$, we have $\LL_{Z_n}(\eta_n(\theta_0))\rightarrow \LL_{Z}(\eta(\theta_0))$. This and Lemma \ref{2.8} then implies $Q(\eta_n(\theta_0))\rightarrow Q(\eta(\theta_0))$. If $\eta(\theta_0)<0$ then we have $\eta_n(\theta_0)<0$ for all $n$ that are sufficiently large. Therefore in this case the sequence $\eta_n(\theta_0)$ and $\eta(\theta_0)$ satisfy the conditions of Lemma \ref{KeyLemma}. Therefore we still have $\LL_{Z_n}(\eta_n(\theta_0))\rightarrow \LL_{Z}(\eta(\theta_0))$ and thus $Q(\eta_n(\theta_0))\rightarrow Q(\eta(\theta_0))$. 

Now if $q=-\tha$ then since $Q_n(\theta)\rightarrow Q(\theta)$ for any $\theta \in (-\tha, -\sqrt{\A/\C})$ and $Q(\theta)\rightarrow +\infty$ when $\theta$ converges to $-\tha$ from the right, we conclude that $Q_n(q_n)\rightarrow +\infty$. But $q_n$ is the minimizing point of $Q_n(\theta)$ in $(-\tha_n, \tha_n)$, a contradiction. Therefore we assume $q\in (-\tha, 0]$ and $q\neq q_{min}$ below.
First consider the case $q\neq \tilde{\theta}$. Then we have $Q_n(q_n)\rightarrow Q(q)>Q(q_{min})$ (here we don't rule out the case $q_{min}=\tilde{\theta}$). As $Q(\theta)$ is a continuous function on $(-\tha, 0)$ and $q_{min}\in (-\tha, 0)$, $q_{min}$ has a small neighborhood $(q_{min}-\delta, q_{min}+\delta)$ with some $\delta>0$ such that for all $\theta \in (q_{min}-\delta, q_{min}+\delta)$ we have $Q(\theta)<Q(q)$. This contradicts with the fact that $Q_n(\theta)\rightarrow Q(\theta)$ for all $\theta \in (q_{min}-\delta, q_{min}+\delta)$ except possibly $\theta=q_{min}=\tilde{\theta}$. Now assume $q=\tilde{\theta}$. We have $Q_n(q_n)=e^{\C_nq_n}Ee^{(\frac{1}{2}\C_nq_n-\frac{1}{2}\A_n)Z_n}<Q_n(\theta)$ for all $\theta \in (-\tha_n, \tha_n)$. When $q_n\rightarrow q=-\frac{\A}{\C}$ we have $\frac{1}{2}\C_nq_n-\frac{1}{2}\A_n\rightarrow 0$. Therefore the random variables $(\frac{1}{2}\C_nq_n-\frac{1}{2}\A_n)Z_n$ converge almost surely to zero. We apply Fatou's lemma to the expression of $Q_n(q_n)$ above and obtain 
\[
Q(q)=e^{\C q}\le \liminf_{n}Q_n(q_n)\le \liminf_{n}Q_n(\theta)=Q(\theta), \; \; \forall \theta \in \bar{\Theta}.
\]
This implies $Q(q)\le Q(q_{min})$ and this contradicts with $Q(q_{min})<Q(q)$.
\vspace{0.1in}

Case 2: Assume $\s=0$. In this case $\tha=\sqrt{\frac{\A}{\C}}$ and $0\le \eta(\theta)\le \frac{1}{2}\A$ for $\theta \in [-\tha, 0]$. We have $\eta(\theta)>0$ for all $\theta\in (-\tha, 0]$. Therefore $Q_n(\theta)\rightarrow Q(\theta)$ for all $\theta\in (-\tha, 0]$.  If $q\in (-\tha, 0]$ then we have $Q_n(q_n)\rightarrow Q(q)>Q(q_{min})$. This contradicts with the fact that $Q_n(\theta)$ converges to $Q(\theta)$ in the neighbourhood of $q_{min}\in (-\tha, 0)$. Now assume $q=-\tha$. We have $Q_n(q_n)=e^{\C_nq_n}Ee^{(\frac{1}{2}\C_nq_n-\frac{1}{2}\A_n)Z_n}$ and when $q_n\rightarrow -\tha$ we have $\frac{1}{2}\C_nq_n-\frac{1}{2}\A_n\rightarrow 0$. Therefore $(\frac{1}{2}\C_nq_n-\frac{1}{2}\A_n)Z_n$ converge almost surely to zero. Again we apply Fatou's lemma and obtain
\[
Q(-\tha)=e^{-\C \tha}\le \liminf_nQ_n(q_n)\le \liminf_nQ_n(\theta)=Q(\theta), \; \forall \theta \in (-\tha, 0].
\]
This implies $Q(-\tha)\le Q(q_{min})$. But $Q(q_{min})<Q(-\tha)$ due to strict convexity of $Q$
in $[-\tha, 0]$ (note here that $\LL_Z(\s)<\infty$).

\end{proof}

\bibliographystyle{plainnat}

\bibliography{main}

\begin{thebibliography}{18}
\providecommand{\natexlab}[1]{#1}
\providecommand{\url}[1]{\texttt{#1}}
\expandafter\ifx\csname urlstyle\endcsname\relax
  \providecommand{\doi}[1]{doi: #1}\else
  \providecommand{\doi}{doi: \begingroup \urlstyle{rm}\Url}\fi

\bibitem[Bondesson(1979)]{Bondesson-scandinal-actuarial-j}
L.~Bondesson.
\newblock On the generalized gamma and generalized negative binomial convolutions.part 1.
\newblock \emph{Scandinavian Actuarial Journal}, 2-3:\penalty0 125--146, 1979.

\bibitem[Bondesson(1988)]{Bondesson-ProbTheory}
L.~Bondesson.
\newblock A remarkable property of generalized gamma convolutions.
\newblock \emph{Probability Theory and Related Fields}, 78:\penalty0 321--333, 1988.

\bibitem[Bondesson(1990)]{Bondesson-1990}
L.~Bondesson.
\newblock Generalized gamma convolutions and complete monotonicity.
\newblock \emph{Probab. theory related fields}, 85 (2):\penalty0 181--194, 1990.

\bibitem[Bondesson(1992)]{Bondesson-lecturenotes}
L.~Bondesson.
\newblock Generalized gamma convolutions and related classes of distributions and densities.
\newblock 1992.

\bibitem[Bondesson(2013)]{Bondesson-TheProb}
L.~Bondesson.
\newblock A class of probability distributions that is closed with respect to addition as well as multiplication of independent random variables.
\newblock \emph{Journal of Theoretical Probability}, 28\penalty0 (3):\penalty0 1063–1081, 2013.

\bibitem[Bosch and Simon(2016)]{Bosch-Simon-2016}
P.~Bosch and T.~Simon.
\newblock A proof of bondesson's conjecture on stable densities.
\newblock \emph{Ark. Math}, 54:\penalty0 31--38, 2016.

\bibitem[Halgreen(1979)]{Halgreen-1979}
C.~Halgreen.
\newblock Self-decomposability of the generalized inverse gaussian and hyperbolic distributions.
\newblock \emph{Z. Wahrscheinlichkeitstheorie verw. Gebiete}, 47:\penalty0 13--17, 1979.

\bibitem[Hammerstein(2010)]{Hammerstein_EAv_2010}
EAv. Hammerstein.
\newblock \emph{Generalized hyperbolic distributions: theory and applications to CDO pricing}.
\newblock PhD thesis, Citeseer, 2010.

\bibitem[James et~al.(2008)James, Boynette, and Yor]{Mark-Yor-GGC}
L.~F. James, B.~Boynette, and M.~Yor.
\newblock Generalized gamma convolutions, dirichlet means, thorin measures, with explicit examples.
\newblock \emph{Probability Survey}, 5:\penalty0 346--415, 2008.

\bibitem[Moschopoulos(1985)]{MOSCHOPOULOS}
P.~G. Moschopoulos.
\newblock The distribution of the sum of independent gamma random variables.
\newblock \emph{Ann. Inst. Statist. Math}, 37:\penalty0 541--544, 1985.

\bibitem[Mukherjea et~al.(2006)Mukherjea, Rao, and Suen]{rao}
A.~Mukherjea, M.~Rao, and S.~Suen.
\newblock A note on moment generating functions.
\newblock \emph{Statist. Peobab. Lett}, 76:\penalty0 1185--1189, 2006.

\bibitem[Nielsen and Halgreen(1977)]{Barndorff-Halgreen-1977}
O.~B. Nielsen and C.~Halgreen.
\newblock Infinite divisibility of the hyperbolic and generalized inverse gaussian distributions.
\newblock \emph{Z. Wahrscheinlichkeitstheorie verw. Gebiete}, 38:\penalty0 309--312, 1977.

\bibitem[Nourdin and Poly(2016)]{Nourdin}
I.~Nourdin and G.~Poly.
\newblock Convergence in law implies convergence in total variation for polynomials in independent gaussian, gamma or beta random variables.
\newblock \emph{In: Houdré, C., Mason, D., Reynaud-Bouret, P., Rosiński, J. (eds) High Dimensional Probability VII. Progress in Probability}, 71\penalty0 (2):\penalty0 382--394, 2016.

\bibitem[Rasony and Sayit(2023)]{rasonyi-hasan}
M.~Rasony and H.~Sayit.
\newblock Exponential utility maximization in small/large financial markets.
\newblock \emph{Preprint}, 2023.

\bibitem[Reis(1989)]{Reis}
R.~Reis.
\newblock Approximate distributions of order statistics, with applications to nonparametric statistics.
\newblock \emph{(Springer, New York.}, 1989.

\bibitem[Stewart(1969)]{Stewart1969OnTC}
G.~W. Stewart.
\newblock On the continuity of the generalized inverse.
\newblock \emph{Siam Journal on Applied Mathematics}, 17:\penalty0 33--45, 1969.

\bibitem[Thorin(1977)]{Thorin1977OnTI}
O.~Thorin.
\newblock On the infinite divisibility of the lognormal distribution.
\newblock \emph{Scandinavian Actuarial Journal}, pages 121--148, 1977.
\newblock URL \url{https://api.semanticscholar.org/CorpusID:120983637}.

\bibitem[Wissem and Thomas(2013)]{Wissem-Thomas}
J.~Wissem and S.~Thomas.
\newblock Further examples of ggc and hcm densities.
\newblock \emph{Bernoulli}, 19\penalty0 (5A):\penalty0 1818--1838, 2013.

\end{thebibliography}

\vspace{0.2in}

\end{document}